\documentclass[10pt,reqno]{amsart}
\usepackage{latexsym}
\usepackage{amsmath}
\usepackage{cases}
\usepackage{amssymb, amscd, amsthm, multirow}
\usepackage[all]{xy}
\usepackage{comment}
\usepackage{bbm}
\usepackage{mathrsfs}
\usepackage{booktabs}
\usepackage{setspace}

\usepackage{threeparttable}
\usepackage{float}
\usepackage{color}

\usepackage{graphicx}

\newtheorem{corollary}{{Corollary}}
\newtheorem{theorem}{{Theorem}}
\newtheorem{definition}{{Definition}}
\newtheorem{lemma}{{Lemma}}
\newtheorem{remark}{{Remark}}
\newtheorem{proposition}{{Proposition}}

\newcommand{\bC}{\mathbb{C}}
\newcommand{\F}{\mathbb{F}}
\newcommand{\Q}{\mathbb{Q}}
\newcommand{\Z}{\mathbb{Z}}

\newcommand{\J}{\mathcal{J}}

\newcommand{\N}{\mathcal{N}}
\DeclareMathOperator{\Ker}{Ker}

\DeclareMathOperator{\Tr}{Tr}
\let\Im\relax\DeclareMathOperator{\Im}{Im}
\DeclareMathOperator{\wt}{wt}
\newcommand{\wF}{\widehat{\F_{2^n}^*}}
\newcommand{\ceil}[1]{\biggl\lceil#1\biggr\rceil}

\usepackage[shortlabels, inline]{enumitem}
\SetEnumerateShortLabel{a}{\textup{(\alph*)}}
\SetEnumerateShortLabel{A}{\textup{(\Alph*)}}
\SetEnumerateShortLabel{1}{\textup{(\arabic*)}}
\SetEnumerateShortLabel{i}{\textup{(\roman*)}}
\SetEnumerateShortLabel{I}{\textup{(\Roman*)}}

\begin{document}

\title[Binomial vectorial functions]{On quadratic binomial vectorial functions with maximal bent components}

\author{Xianhong Xie$^{1}$}
\address{$^1$School of Artificial Intelligence, Anhui Agricultural University, Hefei 230036, China}
\email{xianhxie@ahau.edu.cn}

\author{Yi Ouyang$^{2,3}$}
\address{$^2$School of Mathematical Sciences, University of Science and Technology of China, Hefei 230026, China}
\address{$^3$Hefei National Laboratory, Hefei 230088, China}
\email{yiouyang@ustc.edu.cn}

\author{Shenxing Zhang$^{4,5}$}
\address{$^4$School of Mathematics, Hefei University of Technology, Hefei 230601, China}
\address{$^5$State Key Laboratory of Cyberspace Security Defense, Institute of Information Engineering, Chinese Academy of Sciences, Beijing 100085, China}
%\email{hghu2005@ustc.edu.cn}
\email{zhangshenxing@hfut.edu.cn}
\thanks{Partially supported by NSFC (Grant No. 62402004), QNMP (Grant No. 2021ZD0302902), and State Key Laboratory of Cyberspace Security Defense (Grant No. 2025-MS-04)}

\begin{abstract}
Assume $n=2m\geq 2$ and let $F(x)=x^{d_1}+x^{d_2}$ be a binomial vectorial function over $\F_{2^n}$ possessing the maximal number (i.e. $2^n-2^m$) of bent components. Suppose the $2$-adic Hamming weights $\wt_2(d_1)$ and $\wt_2(d_2)$ are both at most $2$, we prove that $F(x)$ is affine equivalent to either $x^{2^m+1}$ or $x^{2^i}(x+x^{2^m})$, provided that
\[
	\ell(n):=\min_{\gamma:~\F_2(\gamma)=\F_{2^n}} \dim_{\F_2}\F_2[\sigma]\gamma >m,
\]
where  $\sigma$ is the Frobenius $(x\mapsto x^2)$ on $\F_{2^n}$, and $\gcd(d_1,d_2,2^m-1)>1$.
Under this condition, we also establish two bounds on the nonlinearity and the differential uniformity of $F$ by means of the cardinality of its image set.
\end{abstract}

\maketitle

\section{Introduction}%

In this paper, we fix a positive even integer $n=2m$.
Let $\F_{2^i}$ be the finite field of $2^i$ elements, $\sigma$ be the  Frobenius $(x\mapsto x^2)$ on $\F_{2^i}$ and $\Tr_{2^i/2}$ be the trace function from $\F_{2^i}$ to $\F_2$.
For a finite set $X$, we denote by $\#X$ its cardinality.

Suppose that $F: \F_{2^n}\rightarrow \F_{2^n}$ is a vectorial function.
For $a\in \F_{2^n}$, let $F_a$ be the component function
\begin{align*}
	F_a: \F_{2^n}\longrightarrow \F_2,\quad
	v\longmapsto \Tr_{2^n/2}\bigl(a\cdot F(v)\bigr).
\end{align*}
Let
\[
	S_F=\{a\in\F_{2^n}: F_a\ \text{is not bent}\}.
\]
In 2018, Pott et al. \cite{PPMB2018} proved that the cardinality $\#S_F\geq 2^m$, i.e., the number of bent components is at most $2^n-2^m$. Further results about vectorial functions and their bent components were obtained by \cite{anb,hu,mzhang,zheng,xie} and others.

Two fundamental questions arise in this subject: first, to classify all functions $F(x)$ that attain the maximal number of bent components; and second, to clarify key cryptographic properties---such as nonlinearity and differential uniformity---of such functions.

Research on functions attaining the maximal number of bent components has seen notable progress.
In 2023, Hu et al.\ \cite{hu} established that the monomials $x^{2^i(2^m+1)}$ are the only monomials over $\F_{2^n}$ possessing $2^n-2^m$ bent components.
For binomial vectorial functions of the form $F(x)=x^{d_1}+x^{d_2}$, Pott et al.\ \cite{PPMB2018} demonstrated that $F(x)$ achieves this bound when $d_1=2^i+1$ and $d_2=2^m+2^i\ (0\leq i\leq m-1)$.
More recently, Xie et al.\ \cite{xie} refined this analysis for the specific family $F(x)=x^{2^i+1}+x^{2^m+1}\ (0\leq i\leq m-1) $, proving it meets the bound if and only if $i=0$.
Furthermore, for general exponents $(d_1,d_2)$, computational evidence using Magma presented in \cite{xie} suggests that any such binomial $F(x)$ is affine equivalent to either $x^{2^m+1}$ or $x^{2^i}(x+x^{2^m})$ in order to have the maximal number of bent components.
However, a proof of this assertion remains elusive.

Anbar et al.\ \cite{anb} proved that the nonlinearity of a plateaued function with maximal number of bent components is at most $2^{n-1}-2^{\lfloor\frac{3n}{4}\rfloor}$. Xie et al.\ \cite{xie} presented examples of plateaued and non-plateaued vectorial functions attaining the upper bound. Whether this bound holds for non-plateaued functions remains an open question.

The main goal of this paper is to investigate properties of binomial functions $F(x)=x^{d_1}+x^{d_2}$ with maximal bent components, using Stickelberger's Theorem as developed in \cite{Ver, LL, phi, hu}.
Assuming that
\[
\ell(n):=\min_{\gamma:\,\F_2(\gamma)=\F_{2^n}}
\dim_{\F_2}\F_2[\sigma]\gamma>m, \quad  \gcd(d_1,d_2,2^m-1)>1,
\]
our contributions can be summarized as follows.
\smallskip

(A).
We show that $S_F=\F_{2^m}$ and $d_2-d_1\equiv 0\pmod{2^m-1}$ whenever $F$ has $2^n-2^m$ bent components. Moreover, $F$ is equivalent to $x^{2^m+1}$ if $\wt_2(d_1)=1$, and to $x^{2^i+1}+x^{2^i+2^m}$ if $\wt_2(d_1)=\wt_2(d_2)=2$. This generalizes the results in \cite[Theorem~9]{xie} and \cite[Theorem~4.16]{zheng}.
\smallskip

(B).
Using the cardinality of the image set, we establish theoretical bounds on the nonlinearity and differential uniformity of general functions $F$ over $\F_{2^n}$ satisfying $\ell(n)>m$, $\#S_F=2^m$, and $\sigma\circ F=F\circ \sigma$. If in addition $F(x)=x^{d_1}+x^{d_2}$, we prove that
\[
\#\Im(F)=\frac{(2^m-1)c}{s}+1,
\]
and
\begin{align*}
	\N_F&\leq 2^{n-1}-\frac{1}{2}\sqrt{\frac{2^{3m}}{T}\biggl(\frac{2^{3m}s}{s+(2^m-1)c}-2^{m+1}+1\biggr)},\\
	\delta_F&\geq
	\ceil{
		\frac{2^n}{\#\Delta}\biggl(\frac{2^ns}{s+(2^^m-1)c}-1\biggr)
	},
\end{align*}
where $\alpha$ is a generator of $\F_{2^n}^*$,
\[
s=\gcd(d_1,d_2,2^m-1),\qquad
c=\#\bigl\{F(\alpha^i)^{(2^m-1)/s}: i=1,2,\dots,2^m\bigr\},
\]
and
\[
\Delta=\{x+y : (x,y)\in\F_{2^n}\times\F_{2^n},\,x\neq y,\,F(x)=F(y)\}.
\]
\smallskip

The paper is organized as follows.
In \S~2, we provide some necessary preliminaries.
In \S~3, we establish further properties of binomial vectorial functions with the maximal number of bent components.
In \S~4, we derive equivalent forms of $F(x)$ for the cases where $\wt_2(d_1)=1$ or $\wt_2(d_1)=\wt_2(d_2)=2$.
In \S~5, we determine the cardinality of the image set of $F$, and then use it to derive bounds on its nonlinearity and differential uniformity.
In \S~6, we conclude the paper.

\section{Preliminaries}
We always denote $N=2^n-1$.

\subsection{The components of vectorial functions}
We recall some facts about vectorial functions.

\begin{definition}\label{def:vectorial-function}
	Suppose that $F: \F_{2^n}\rightarrow \F_{2^n}$ is a vectorial function.
	\begin{enumerate}[i]
		\item The component function $F_a$ of $F$ at $a\in \F_{2^n}$ is the Boolean function
		\begin{align*}
			F_a: \F_{2^n}&\longrightarrow \F_2\\
			v&\longmapsto \Tr_{2^n/2}\bigl(a\cdot F(v)\bigr).
		\end{align*}
		\item The Walsh transform of $F$ is
		\begin{equation*}
			W_F(a, \omega):=W_{F_a}(\omega)=\sum_{v\in \F_{2^n}}(-1)^{F_a(v)+\Tr_{2^n/2}(\omega v)},
		\end{equation*}
		where $a\in \F_{2^n}^*$, $\omega \in \F_{2^n}$.
		\item If $W_{F_a}(\omega)=\pm2^{\frac{n}{2}}$ for all $\omega\in\F_{2^n}$, then $F_a$ is called a bent component of $F$.
		Denote
		\[
			S_F:=\{a\in \F_{2^n}: F_a \text{ is not bent}\}.
		\]
		\item If $W_{F_a}(\omega)\in\{0,\pm2^{\frac{n+k}{2}}\}$ for all $\omega\in\F_{2^n}$, where $k\in\Z$ and $k\equiv n\bmod2$, then $F_a$ is called a $k$-plateaued component of $F$.
		If $F_a$ ($a\in\F_{2^n}^*$) are all plateaued, then $F$ is called plateaued.
		\item The nonlinearity of $F$ is the minimal nonlinearity among its component functions, i.e.,
		\[
			\N_F:=2^{n-1}-\frac{1}{2}\max_{a\in\F_{2^n}^*, \omega\in\F_{2^n}}|W_{F_a}(\omega)|.
		\]
		\item The differential uniformity of $F$ is
		\[
			\delta_F:=\max_{a\in\F_{2^n}^*,b\in\F_{2^n}}\#\{x\in\F_{2^n}: F(x+a)+F(x)=b\}.
		\]
	\end{enumerate}
\end{definition}

Pott et al.\ \cite{PPMB2018} proved the following result.

\begin{theorem}\label{thm:SF-vector-space}
	$\#S_F\geq 2^m$.
	Moreover, $\#S_F=2^m$ if and only if $S_F$ is an $m$-dimensional $\F_2$-subspace of $\F_{2^n}$.
\end{theorem}

We say that $F$ has maximal number of bent components if $\#S_F= 2^m$, i.e. the number of bent components of
	$F$ is exactly
$2^n-2^m$. A natural problem is to determine all vectorial functions with maximal number of bent components.

\begin{theorem}\label{thm:hu}
	If $F(x)=x^{d_1}$ has maximal number of bent components, then
	\begin{enumerate}[1]
		\item $($Zheng et al.\ \cite{zheng}$)$ $S_F=\F_{2^m}$ and $(2^m+1)\mid d_1$;
		\item $($Hu et al.\ \cite{hu}$)$ $d_1=(2^m+1)s$, where $s\in\{1,2,2^2,\ldots,2^{m-1}\}$.
	\end{enumerate}
\end{theorem}

For the nonlinearity of plateaued functions with maximal number of bent components, we know
\begin{theorem}\label{platea}
	If $F$ is a plateaued function and $\#S_F=2^m$, then
	\[
		\N_F\leq2^{n-1}-2^{\lfloor\frac{3n}{4}\rfloor}.
	\]
\end{theorem}

\subsection{Gauss sums and Stickelberger's Theorem}

Let $\Q_2$ denote the field of $2$-adic numbers and $\overline{\Q}_2$ be a fixed algebraic closure of $\Q_2$.
We regard $\overline{\Q}_2$ as a subfield of $\bC$, the field of complex numbers.

Let $\xi$ be a primitive $N$-th root of unity in $\overline{\Q}_2$ (recall $N=2^n-1$).
The algebraic extension $\Q_2(\xi)/\Q_2$ is unramified of degree $n$.
We identify $\F_{2^n}$ with the residue field of $\Q_2(\xi)$ that is $\Z_2[\xi]/(2)$.
For any $a\in \F_{2^n}$, there exists a canonical lifting $\omega(a) \in \Z_2[\xi]$ such that $\omega(ab)=\omega(a)\omega(b)$ and
\begin{equation} \label{eq:tech-lift}
	\omega(a)\bmod2=a,\ a\in\F_{2^n}.
\end{equation}
The character $\omega$ is called the Teichm\"{u}ller character of $\F_{2^n}$.
The group $\wF$ of multiplicative characters of $\F_{2^n}^*$ is a cyclic group of order $N$ generated by $\omega$:
\[
	\wF=\{\chi:\ \F_{2^n} \rightarrow \bC \}
	=\{\omega^j:\ 0\leq j\leq N-1\}.
\]
As is customary, the definition domain of a character $\chi$ is extended to $\F_{2^n}$ by setting $\chi(0)=0$ if $\chi$ is nontrivial and $\chi(0)=1$ if $\chi=\omega^0$ is trivial.

\begin{definition}
	For $\chi\in\wF$, the Gauss sum $G(\chi)$ over $\F_{2^n}$ is defined by
	\[
		G(\chi)=\sum_{x\in\F^*_{2^n}}\psi(x)\chi(x),
	\]
	where $\psi(x)=(-1)^{\Tr_{2^n/2}(x)}$ is the canonical additive character of $\F_{2^n}$.
\end{definition}

Obviously, $G(\chi)=-1$ if $\chi$ is trivial, and $G(\chi) G(\chi^{-1})=\chi(-1) 2^n$ if $\chi$ is nontrivial.
For any $x\in\F_{2^n}^*$, applying the inverse Fourier transform, we obtain
\begin{equation}\label{eq:fourier}
	\psi(x)=\frac{1}{N}\sum_{\chi\in\wF}G(\chi) {\chi^{-1}}(x)
	=\frac{1}{N}\sum_{j=0}^{N-1} G(\omega^{-j})\omega^{j}(x).
\end{equation}

We identify $\Z_N$ with the set $\{0,1,\cdots, N-1\}$.
For any integer $j$, let $j_N=j\bmod{N}\in\Z_N$ be its minimal non-negative residue.
For any positive integer $j$, let $\wt_2(j)$ denote its $2$-adic Hamming weight, namely
\[
	\wt_2(j)=\wt_2(j_N)=\sum_{i=0}^{n-1} c_i,\ \text{where}\
	j_N=\sum_{i=0}^{n-1} c_i2^i
\]
is the binary representation of $j_N=j\bmod{N}$.

Let $\bar{b}$ denote the bit-complement of $b\in\{0,1\}$, i.e., $\bar{b}=1-b$. Then
\[
(-j)_N=\sum_{i=0}^{n-1}\bar{c_i}2^i\quad \text{and}\quad \wt_2(-j)=\sum_{i=0}^{n-1}\bar{c_i}=n- \wt_2(j).
\]

The following theorem of Stickelberger is a well-known result in algebraic number theory and is very useful for analyzing the exponential sums.

\begin{theorem}[\cite{Sti}]\label{stick}
	For any $0\leq i<N$,
	\[
		G(\omega^{-i})\equiv2^{\wt_2(i)}\mod{2^{\wt_2(i)+1}}.
	\]
\end{theorem}

\section{Binomial vectorial functions with maximal bent components}

From now on, let
\[
	F(x)=x^{d_1}+x^{d_2}:\F_{2^n}\rightarrow\F_{2^n}
\]
be a binomial vectorial function, where $d_1,d_2\in \Z_N\setminus\{0\}$.
Our goal is to characterize the functions $F(x)$ for which the number of bent components is maximal.

\subsection{Known results}

The following result is due to Pott et al.~\cite{PPMB2018}.

\begin{theorem} \label{thm:pott}
	The function $F(x)=x^{2^i+1}+x^{2^i+2^m}\ (0<i<m)$ has maximal number of bent components.
\end{theorem}

Two vectorial functions $F,F':\F_{2^n}\rightarrow\F_{2^n}$ are called  EA-equivalent if there exist affine automorphisms $A,A':\F_{2^n}\rightarrow\F_{2^n}$ and an affine function $L:\F_{2^n}\rightarrow\F_{2^n}$ such that $F'=A'\circ F\circ A+L$. Note that the property of having maximal number of bent components is invariant under EA-equivalent.

 By Theorem~\ref{thm:pott},  if replacing $i$ by $m-i$, then $F(x)=x^{1+2^{m-i}}+x^{2^{m-i}+2^m}\ (0<i<m)$ has maximal number of bent components. Set
\[A(x)=x^{2^{m+i}},\ A'=x,\ \text{and}\ L=0.\]
Then we have
\begin{corollary}\label{cor:pott}
  The function $F(x)=x^{1+2^i}+x^{1+2^{m+i}}$ has maximal number of bent components.
\end{corollary}

Moreover, if replacing $2^i+2^m$ by $2^{m+i}+2^m$, then
\begin{proposition}\label{pro:qua}
	The number of bent components of $F(x)=x^{2^i+1}+x^{2^{m+i}+2^m}\ (0\leq i\leq m)$ is not maximal.
\end{proposition}

\begin{proof}
	This is trivial if $i=0$ or $m$. Now we suppose $0<i<m$.
	
	For $a\in\F_{2^n}^*$, $F_a(x)$ is a quadratic form, whose associated bilinear form is
	\begin{align*}
		B_a(x,z)&=F_a(x+z)+F_a(x)+F_a(z)\\
		&=\Tr_{2^n/2}\bigl(aF(x+z)\bigr)+\Tr_{2^n/2}\bigl(aF(x)\bigr)+\Tr_{2^n/2}\bigl(aF(z)\bigr)\\
		&=\Tr_{2^n/2}\bigl(ax^{2^{i}}z+axz^{2^{i}}+
		ax^{2^m}z^{2^{m+i}}+ax^{2^{m+i}}z^{2^m}\bigr)\\
		&=\Tr_{2^n/2}\bigl(z(ax^{2^{i}}+(ax)^{2^{n-i}}+a^{2^{m-i}}x^{2^{n-i}}+a^{2^m} x^{2^{i}})\bigr).
	\end{align*}
	Let $L_a(x):=ax^{2^{i}}+(ax)^{2^{n-i}}+a^{2^{m-i}}x^{2^{n-i}}+a^{2^m}
	x^{2^{i}}$. By a result of Hu and Feng \cite{Hufeng}, $F_a(x)$ is bent if and only if $F_a(x)$ is non-degenerate. As $n$ is even, this is equivalent to that $B_a(x,z)$ is non-degenerate.
	As the trace map is non-degenerate, we get
	\[
		F_a\ \text{is not bent}
		\iff \exists\ x\neq 0,\ \text{such that}\ L_a(x)=0.
	\]
	We conclude that $F$ has more than $2^m$ non-bent components:
	\begin{enumerate}[i]
		\item If $a\in\F_{2^m}$, then $L_a(x)=0$ for any $x\in\F_{2^n}$, so $F_a$ is not bent.
		\item The trace map $\Tr_{2^n/2^m}$ is surjective, so there exists $a\in \F_{2^n}\setminus\F_{2^m}$ such that $\Tr_{2^n/2^m}(a)=a+a^{2^m}=1$. For such $a$, $L_a(1)=0$, hence $F_a$ is not bent.\qedhere
	\end{enumerate}
\end{proof}

\subsection{The study of $\J_{d_1,d_2}$}
For a positive integer $i$, we denote by $v_2(i)$ its $2$-adic valuation.

\begin{definition}
	Let $\J= \Z_{N}\times\Z_{N}\setminus\{(0,0)\}$.
	Denote
	\begin{align}
		V_{d_1,d_2} &:=\Bigl((j_1, j_2)\in \J\mapsto \wt_2(j_1)+\wt_2(j_2)+\wt_2(-d_1j_1-d_2j_2)\Bigr); \notag\\
		\nu_{d_1,d_2}&:=\min\{V_{d_1,d_2}(j_1,j_2) : (j_1,j_2)\in \J\};\notag\\
		\J_{d_1,d_2}&:=\{(j_1,j_2)\in \J :
		V_{d_1,d_2}(j_1,j_2)=\nu_{d_1,d_2}\}.\label{eq:def-J}
	\end{align}	
\end{definition}

Obviously, the set $\J_{d_1,d_2}$ is closed under multiplication by $2$.

\begin{theorem} \label{thm:th3}
	\begin{enumerate}[1]
		\item For any $(a,b)\in\F_{2^n}^*\times\F_{2^n}$, $v_2\bigl(W_{F_a}(b)\bigr)\geq \nu_{d_1,d_2}$.
		\item Denote the polynomial
		\[
			g_a(x)=\sum_{(j_1,j_2)\in\J_{d_1,d_2}}a^{j_1+j_2}x^{(-d_1j_1-d_2j_2)_N}.
		\]
		Then $g_a(b)\in \F_2$ and
		\begin{equation}
			v_2\bigl(W_{F_a}(b)\bigr)>\nu_{d_1,d_2}
			\iff g_a(b)=0.
		\end{equation}
	\end{enumerate}
\end{theorem}

\begin{proof}
	For $(a,b)\in\F_{2^n}^*\times\F_{2^n}$, we have
	\begin{align*}
		W_{F_a}(b)&=1+\sum_{x\in\F_{2^n}^*}(-1)^{\Tr_{2^n/2}(ax^{d_1}+ax^{d_2}+bx)}\\
		&=1+\sum_{x\in\F_{2^n}^*}\psi(ax^{d_1})\psi(ax^{d_2})\psi(bx)=1+\frac{S}{N^3},
	\end{align*}
	where $S$ is given as follows by Eq.~\eqref{eq:fourier}:
	\begin{align*}
		S=& \sum_{x\in\F_{2^n}^*}\sum_{j_1,j_2,j_3=0}^{N-1} G(\omega^{-j_1})G(\omega^{-j_2})G(\omega^{-j_3})\omega^{j_1}(ax^{d_1})
		\omega^{j_2}(ax^{d_2})\omega^{j_3}(bx) \\
		=& \sum_{j_1,j_2,j_3=0}^{N-1} G(\omega^{-j_1})G(\omega^{-j_2})G(\omega^{-j_3}) \omega^{j_1+j_2}(a) \omega^{j_3}(b) \sum_{x\in\F_{2^n}^*} \omega^{d_1j_1+d_2 j_2+j_3}(x)\\
		=& N\, \sum_{j_1,j_2=0}^{N-1} G(\omega^{-j_1})G(\omega^{-j_2})G(\omega^{d_1j_1+d_2 j_2}) \omega^{j_1+j_2}(a) \omega^{-d_1j_1 -d_2 j_2}(b)\\
		=& N\,
		\sum_{(j_1,j_2)\in \J} G(\omega^{-j_1})G(\omega^{-j_2})G(\omega^{d_1j_1+d_2 j_2}) \omega(a^{j_1+j_2} b^{-d_1j_1 -d_2 j_2}) -N.
	\end{align*}
	Recall that $N=2^n-1$, one have
	\[
		W_{F_a}(b)\equiv \sum_{(j_1,j_2)\in \J} G(\omega^{-j_1})G(\omega^{-j_2})G(\omega^{d_1j_1+d_2 j_2}) \omega(a^{j_1+j_2} b^{-d_1j_1 -d_2 j_2})\mod{2^n}.
	\]
	By Theorem~\ref{stick} and Eq.~\eqref{eq:def-J}, we get
	\[
		W_{F_a}(b)
		= 2^{\nu_{d_1,d_2}}\sum_{(j_1,j_2)\in\J_{d_1,d_2}}\omega
		(a^{j_1+j_2}b^{-(d_1j_1+d_2j_2)})\mod{2^{\nu_{d_1,d_2}+1}}.
	\]
	Thus $v_2\bigl(W_{F_a}(b)\bigr)\geq \nu_{d_1,d_2}$.
	
	Since $\J_{d_1,d_2}$ is closed under multiplication by $2$, we have $g_a^2(b)=g_a(b)$ and $g_a(b)\in\F_2$. By Eq.\eqref{eq:tech-lift},
	\[
		g_a(b)\equiv \sum_{(j_1,j_2)\in\J_{d_1,d_2}}
		\omega(a^{j_1+j_2}b^{-(d_1j_1+d_2j_2)})\mod{2},
	\]
	and we have
	\[
		v_2\bigl(W_{F_a}(b)\bigr)> \nu_{d_1,d_2}\Longleftrightarrow g_a(b)=0.
	\]
	The proof is complete.
\end{proof}

\begin{remark}
	Analogously, let $F(x)=x^{d_1}+1$, i.e. $d_2=0$. The notions $\nu_{d_1,0}$ and $\J_{d_1}:=\J_{d_1,0}$ were introduced and studied in \cite{Ver, lea}. It was shown there that for
	$(a,b)\in\F_{2^n}^*\times\F_{2^n}$, $v_2\bigl(W_{F_a}(b)\bigr)\geq \nu_{d_1,0}$ and
	\[
		v_2\bigl(W_{F_a}(b)\bigr)> \nu_{d_1,0} \Longleftrightarrow\sum_{j_1\in\J_{d_1}}a^{j_1}b^{-d_1j_1}=0.
	\]
\end{remark}

By \cite[Theorem~13]{Mor}, we have
\begin{lemma}
	For any $(a,b)\in\F_{2^n}\times\F_{2^n}$,
	% $W_{F_a}(b)$ can be divisible by $2^{\lceil\frac{n}{\max\{\wt_2(d_1),\wt_2(d_2)\}}\rceil}$, i.e.,
	\begin{equation} \label{eq:bound}
		v_2\bigl(W_{F_a}(b)\bigr)\ge\nu_{d_1,d_2}
		\geq\ceil{\frac{n}{\max\{\wt_2(d_1),\wt_2(d_2)\}}}.
	\end{equation}
\end{lemma}

\subsection{Auxiliary results}

\begin{lemma}\label{lem:rz}
	Assume $s=\gcd(d_1,d_2,N)> 1 $. If $F$ has at least one bent component, then $s$ is a factor of one of $2^m\pm 1$ and coprime to the other.
	Specifically,
	\[
		s\mid (2^m\pm 1) \iff W_{F_a}(0)=\mp 2^m \ \text{for any}\ a\in\F_{2^n}\setminus S_F.
	\]
\end{lemma}

\begin{proof}
	Let $\alpha$ be a primitive element of $\F_{2^n}$ and $W=\{\gamma\in\F_{2^n} : \gamma^s=1\}$.
	Then
	\[
		W= \langle\alpha^{\frac{N}{s}}\rangle,\qquad \F_{2^n}^*=\bigcup_{i=0}^{\frac{N}{s}-1}\alpha^i W.
	\]
	Clearly, $F_a(x)$ is constant on each coset $\alpha^i W$, consequently
	\begin{align*}
		W_{F_a}(0)&=1+\sum_{x\in\F^*_{2^n}}(-1)^{\Tr_{2^n/2}(ax^{d_1}+ax^{d_2})}\\
		&=1+s\sum_{i=0}^{\frac{N}{s}-1}(-1)^{\Tr_{2^n/2}(a\alpha^{d_1i}+a\alpha^{d_2i})}
		\equiv1\mod{s}.
	\end{align*}
	By definition $W_{F_a}(0)=\pm2^m$ for $a\in\F_{2^n}\setminus S_F$.
	Thus if $\F_{2^n}\setminus S_F\neq \emptyset$, then $s\mid (2^m\pm 1)$. As $\gcd(2^m-1,2^m+1)=1$, $s$ is a factor of one of $2^m\pm 1$ and coprime to the other.
\end{proof}

\begin{lemma}[\cite{LL}, Lemma~2]\label{lem:LL}
	If $0<j<N$, then $\wt_2(j)+\wt_2(-j)=n$.
	Moreover, if $(2^m+1)\nmid j$, then $\wt_2((2^m-1)j)=m$.
\end{lemma}

\begin{lemma}\label{lem:deg}
	Suppose $(j_1,j_2)\in \J$.
	\begin{enumerate}[1]
		\item If $(j_1+j_2)_N\geq2^n-2^m+1$, then
		$V_{d_1,d_2}(j_1,j_2)\geq m+1$.
		\item If $j_1+j_2\equiv 0\mod{(2^m-1)}$, $j_1+j_2\neq N$ and $V_{d_1,d_2}(j_1,j_2)=m$, then
		\[
			\wt_2(j_1)+\wt_2(j_2)=\wt_2(j_1+j_2)=m,\qquad
			d_1j_1+d_2j_2\equiv 0\mod N.
		\]
	\end{enumerate}
\end{lemma}

\begin{proof}
	(1) Write $(j_1+j_2)_N=N-u$. Then $0< u\leq 2^m-2$. Note that $\wt_2(u)\leq m-1$ and $\wt_2(j_1+j_2)\leq\wt_2(j_1)+\wt_2(j_2)$, then
	\[
		\wt_2(j_1)+\wt_2(j_2)
		\geq\wt_2(j_1+j_2)
		=\wt_2(N-u)=n-\wt_2(u)\geq m+1.
	\]
	Hence, $V_{d_1,d_2}(j_1,j_2)\geq m+1+\wt_2(-(d_1j_1+d_2j_2))\geq m+1$.
	
	(2) If $j_1+j_2\equiv 0\mod{(2^m-1)}$ and $j_1+j_2\neq N$, then $\wt_2(j_1+j_2)= m$ by Lemma~\ref{lem:LL}.
	Thus
	\begin{align*}
		V_{d_1,d_2}(j_1,j_2)&=\wt_2(j_1)+\wt_2(j_2)+\wt_2(-(d_1j_1+d_2j_2))\\
		&\geq m+\wt_2(-(d_1j_1+d_2j_2))\geq m.
	\end{align*}
	The equality holds only if
	\[
		d_1j_1+d_2j_2\equiv 0\mod{N},\qquad
		\wt_2(j_1)+\wt_2(j_2)=\wt_2(j_1+j_2)=m.
		\qedhere
	\]
\end{proof}

\begin{lemma}\label{lem:L-gamma}
If $\gamma\in S_F$, then $\sigma(\gamma)=\gamma^2\in S_F$. Hence if $\# S_F=2^m$, then $\F_2[\sigma]\gamma=\{f(\sigma)\gamma\mid  f(x)\in \F_2[x]\}$ is a subspace of $S_F$. 
\end{lemma}

\begin{proof}
	We have
	\begin{align*}
		W_{F_{\gamma^2}}(b)&=\sum_{x\in\F_{2^n}}(-1)^{\Tr_{2^n/2}(\gamma^2F(x)+bx)}\\
		&=\sum_{x\in\F_{2^n}}(-1)^{\Tr_{2^n/2}(\gamma^2F(x^2)+b^{2^n} x^2)}
		=W_{F_{\gamma}}(b^{2^{n-1}}).
	\end{align*}
Thus if $\gamma\in S_F$, then $\gamma^2\in S_F$. If $\# S_F=2^m$, then $S_F$ is a vector space by Theorem~\ref{thm:SF-vector-space}, hence $\F_2[\sigma]\gamma \subseteq S_F$.
\end{proof}

\subsection{The study of $S_F$}

We now give a general result, which also applies to the binomials $x^{d_1}+x^{d_2}$.

For any $\gamma\in \F_{2^n}$, define
\[
	\ell(\gamma)=\dim_{\F_2} \F_2[\sigma]\gamma.
\]
Clearly, $\ell(\gamma)\leq [\F_2(\gamma):\F_2]$, and hence $\ell(\gamma)\leq m$ if $\F_2(\gamma)\neq \F_{2^n}$.
Denote by
\[
	\ell(n):=\min\{\ell(\gamma) :\F_2(\gamma)=\F_{2^n}\}.
\]
By definition, $\ell(n)$ can be regarded as the linear complexity of the Frobenius orbit.  

\begin{theorem}\label{thm:f-x2-SF} 
	Suppose that $F(x)$ is a vectorial function over $\F_{2^n}$ satisfying $\sigma\circ F=F\circ\sigma$ and $\#S_F=2^m$.
	If $\ell(n)>m$, then $S_F=\F_{2^m}$.
\end{theorem}

\begin{proof}
Assume $S_F\neq\F_{2^m}$.
Let $S_1=S_F\setminus\F_{2^m}$ and $S_2=S_F\cap\F_{2^m}$.
Since $S_F$ is a vector space by Theorem~\ref{thm:SF-vector-space}, we have
	\[
		\#S_1\geq1,\quad
		\#S_2\leq2^{m-1},\quad% \text{and}\
		S_F=S_1\cup S_2.
	\]
For any $\gamma\in S_1$, $\F_2(\gamma)\nsubseteq \F_{2^m}$. Note that $ \F_2[\sigma]\gamma \subseteq S_F$ by Lemma~\ref{lem:L-gamma}. If $\gamma$ is a generator of $\F_{2^n}$, then
	\[
		m=\dim_{\F_2}S_F \ge\dim_{\F_2} \F_2[\sigma]\gamma \ge\ell(n)>m,
	\]
	which is impossible.
	Thus $\F_2(\gamma)\neq \F_{2^n}$, $\F_2(\gamma)\subseteq \F_{2^{n/p}}$ for some odd prime factor of $n$.
	This forces that $n$ is not a power of $2$.
	
	Suppose that $p_1<p_2<\cdots<p_l$ are the odd prime factors of $m$.
	Then $\F_2(\gamma)\subseteq\F_{2^{n/p_i}},\gamma\in \F_{2^{n/p_i}}\setminus\F_{2^{m/p_i}}$ for some $i$.
	Thus
	\begin{align*}
		\#S_1&
		\le \sum_{i=1}^l (2^{n/p_i}-2^{m/p_i})
		\le l (2^{n/3}-2^{m/3})\\&
		\le 2^{3^{l-1}-1}(2^{2m/3}-2^{m/3})
		<2^{m/3-1}\cdot 2^{2m/3}=2^{m-1}
	\end{align*}
	and $\#S_F<2^{m-1}+2^{m-1}=2^m$, which contradicts to $\#S_F=2^m$.
	Hence $S_F=\F_{2^m}$.
\end{proof}

\begin{remark}
	Unfortunately there exist examples such that $\ell(n)\leq m$.
	We present experimental results on $\ell(n)$ in Table~\ref{table:t1}.
	It would be very interesting to learn more about $\ell(n)$.
	\begin{table}[H]
		\centering
		\caption{Values of $\ell(n)$}\label{table:t1}
		\begin{tabular}{cccccccc}
			\toprule
			$n$ & $\ell(n)$ && $n$ & $\ell(n)$ && $n$ & $\ell(n)$\\
			\midrule
			4 & 3 && 12 & 5 && 20 & 7 \\
			6 & 4 && 14 & 5 && 22 & 12 \\
			8 & 5 && 16 & 9 && 24 & 7 \\
			10 & 6 && 18 & 8 && 26 & 14 \\
			\bottomrule
		\end{tabular}
	\end{table}
\end{remark}

In the following, we present two cases for which the inequality $\ell(n)>m$ is satisfied.

\begin{proposition}\label{prop:2p}
One has $\ell(n)>m$ if
	\begin{enumerate}[1]
		\item $n=2p$, where $p$ is an odd prime and $2$ is a primitive root modulo $p$, or
		\item $n=2^k$, where $k$ is a positive integer.
	\end{enumerate}
\end{proposition}

\begin{proof}
	(1) Note that
	\[
		x^{2p}-1=(x^p-1)^2
		=(x-1)^2\Phi_p(x)^2,\quad \Phi_p(x):=\dfrac{x^p-1}{x-1}.
	\]
	Denote $\zeta\in\overline{\F}_2$ satisfying $\zeta\ne1$ and $\zeta^p=1$.
	Note that $\zeta^{2^k-1}=1$ if and only if $p\mid (2^k-1)$.
	Since $2$ is a primitive root modulo $p$, we have $\F_2(\zeta)=\F_{2^{p-1}}$ and thus the minimal polynomial $\Phi_p$ of $\zeta$ is irreducible.

	Let $f_\alpha(x)$ be the minimal polynomial of $\sigma$ on $\F_2[\sigma]\alpha$, where $\alpha $ is a generator of $\F_{2^n}$. By linear algebra, $\deg f_{\alpha}= \dim_{\F_2} \F_2[\sigma]\alpha=\ell(\alpha)$. Since $\sigma^{2p}(\alpha)=\alpha$  and $\sigma^p(\alpha)\neq \alpha$, we have
	\[
		f_{\alpha}(x)\mid (x^{2p}-1),\quad f_{\alpha}(x)\nmid (x^p-1),
		\quad f_{\alpha}(x)\nmid (x^2-1).
	\]
Thus either $\Phi_p(x)^2\mid f_{\alpha}$ or $(x-1)^2\Phi_p(x)\mid f_{\alpha}$. This implies that
	\[
		\ell(\alpha)=\deg f_{\alpha}\ge 2+(p-1)=p+1.
	\]
	Hence $\ell(2p)>p$.
	
	(2) 	In this case, we have
	\[
		f_\alpha(x)\mid(x^{2^k}-1)=(x-1)^{2^k},\quad f_\alpha(x)\nmid (x^{2^{k-1}}-1)=(x-1)^{2^{k-1}}.
	\]
Hence  $(x-1)^{2^{k-1}+1}\mid f_\alpha$ and then $\ell(\alpha)\ge 2^{k-1}+1$. Thus $\ell(2^k)>2^{k-1}$.
\end{proof}

\section{The equivalence of quadratic binomial vectorial functions with maximal bent components}

In this section, we always assume
\[
	\ell(n)>m,\qquad
	\max\{\wt_2(d_1),\wt_2(d_2)\}=2,
\]
and denote $s=\gcd(d_1,d_2,N)$.
Note that if $d_1=2^l$ for some integer $l\geq0$, then $F(x)=x^{d_1}+x^{d_2}$ is EA-equivalent to $x^{d_2}$.
By Theorem~\ref{thm:hu}, this implies that $d_2=(2^m+1)2^i$ whenever $S_F=2^m$.
Hence, we proceed to show the equivalence for the case where
\[
	\wt_2(d_1)=\wt_2(d_2)=2.
\]

\begin{theorem}\label{thm:min-wt-is-m}
	Assume that $\#S_F=2^m$ and $\wt_2(d_1)=\wt_2(d_2)=2$. Then $\nu_{d_1,d_2}=m$ and there exists $(j_1,j_2)\in\J_{d_1,d_2}$ such that
	\[
		\wt_2(j_1)+\wt(j_2)=m\ \text{and}\ d_1j_1+d_2j_2\equiv 0\mod{N}.
	\]
	Furthermore, there exists $0\leq j\leq 2^m-1$, such that $(j, 2^m-1-j)\in\J_{d_1,d_2}$, equivalently
	\[
		(d_1-d_2)j + (2^m-1) d_2 \equiv 0\mod{N}.
	\]
\end{theorem}

\begin{proof} Since $\wt_2(d_1)=\wt_2(d_2)=2$, we have $\nu_{d_1,d_2}\geq m$ by Eq.~\eqref{eq:bound}.
	For $a\in\F_{2^n}\setminus S_F$, $W_{F_a}(b)=\pm2^m$ for all $b\in\F_{2^n}$. Hence $\nu_{d_1,d_2}= m$ by Eq.~\eqref{eq:bound} and Theorem~\ref{thm:th3}.
	
	By Theorem~\ref{thm:th3} again, for $a\in\F_{2^n}\setminus S_F$ and $b\in \F_{2^n}$, we have
	\begin{equation}\label{eq:g-a-b-1}
		g_a(b)=\sum_{(j_1,j_2)\in\J_{d_1,d_2}}a^{j_1+j_2}b^{-(d_1j_1+d_2j_2)}=1\in \F_2.
	\end{equation}
	This means that $g_a(x)$ as a polynomial of $x$, whose degree $<N$, has at least $2^n>N$ roots, hence $g_a(x)$ must be the zero polynomial.
	For $t\in \Z_N$, denote
	\[
		J_{d_1,d_2}^t:=\{(j_1,j_2)\in\J_{d_1,d_2} : d_1j_1+d_2j_2\equiv -t\mod{N}\}.
	\]
	Then the coefficient of $x^t$ in $g_a(x)$ is $\sum_{(j_1,j_2)\in J^t_{d_1,d_2}}a^{j_1+j_2}$.
	Hence,
	\[
		\sum_{(j_1,j_2)\in J^t_{d_1,d_2}}a^{j_1+j_2}=0\ (1\leq t<N),\qquad% \text{ and}\
		\sum_{(j_1,j_2)\in J^0_{d_1,d_2}}a^{j_1+j_2}=1.
	\]
	
	Since the non-zero polynomial
	\[
		\sum_{(j_1,j_2)\in J^0_{d_1,d_2}} x^{(j_1+j_2)_N}-1
	\]
	has $ 2^n-2^m$ zeros in $\F_{2^n}\setminus S_F$, so its degree is at least $2^n-2^m$.
	By Lemma~\ref{lem:deg}(1), its degree must be $2^n-2^m$.
	Thus there must have some $(j_1,j_2)\in J_{d_1,d_2}^0$ such that $(j_1+j_2)_N=2^n-2^m$.
	By Lemma~\ref{lem:deg}(2), $d_1j_1+d_2j_2\equiv 0\mod{N}$.
	% \begin{comment}
	% 	hence $(j_1,j_2)\in J^0_{d_1,d_2}$, which is a contradiction. Thus $J^t_{d_1,d_2}=\emptyset\ (1\leq t<N)$ and $\J_{d_1,d_2}=J_{d_1,d_2}^0$. By the same argument for $t=0$, there exists $(j_1,j_2)\in \J_{d_1,d_2}$ such that $(j_1+j_2)_N=2^n-2^m$, thus $\wt_2(j_1)+\wt_2(j_2)\geq \wt_2(j_1+j_2)=m$ and $\nu(d_1,d_2)=V_{d_1,d_2}(j_1,j_2)\geq m$, which contradicts to $\nu_{d_1,d_2}<m$. Hence we must have $\nu_{d_1,d_2}=m$.
		
	% 	\bigskip
	% 	By Eq.~\eqref{q1}, we have
	% 	\begin{align}
	% 		g_a(b)=\sum_{(j_1,j_2)\in\J_{d_1,d_2}}a^{j_1+j_2}b^{-(d_1j_1+d_2j_2)}=1.\label{sti}
	% 	\end{align}
	% 	Again, the polynomial $g_a(x)$ given by \eqref{qw} is the constant polynomial $1$, thus
	% 	\[\sum_{(j_1,j_2)\in J^t_{d_1,d_2}}a^{j_1+j_2}=0\ (1\leq t<N)\ \text{ and}\ \sum_{(j_1,j_2)\in J^0_{d_1,d_2}}a^{j_1+j_2}=1. \]
	% 	By repeating the argument for the polynomials $\sum_{(j_1,j_2)\in J^t_{d_1,d_2}} x^{(j_1+j_2)_N}$, we get $J^t_{d_1,d_2}=\emptyset\ (1\leq t<N)$ and $\J_{d_1,d_2}=J_{d_1,d_2}^0$.
		
	% 	Finally, the polynomial $\sum_{(j_1,j_2)\in J^0_{d_1,d_2}}x^{(j_1+j_2)_N}-1$ has $2^n-2^m$ zeros, so its degree $\geq 2^n-2^m$ and hence must be $2^n-2^m$, i.e. there exists $(j_1,j_2)\in \J_{d_1,d_2}$ such that $(j_1+j_2)_N=2^n-2^m$.
	% \end{comment}
	Suppose
	\[
		j_1=\sum_{a\in S_1} 2^a,\quad j_2=\sum_{a\in S_2} 2^a,\quad j_1+j_2=\sum_{a\in T} 2^a.
	\]
	where $S_1,\ S_2\subseteq \{0,1,\cdots, n-1\}$ and $T\subseteq \{0,1,\cdots, n\}$.
	Then $|T|\leq |S_1|+|S_2|$ with equality if and only if $S_1\cap S_2=\emptyset$.
	By the identity
	\[
		\wt_2(j_1)+\wt_2(j_2)=|S_1|+|S_2|=\wt_2(j_1+j_2)=|T\setminus \{n\}| =m,
	\]
	one has $S_1\cap S_2=\emptyset$, $T=S_1\cup S_2$ and $j_1+j_2=(j_1+j_2)_N= 2^n-2^m$.
	This means that $j_1=2^m j$ and $j_2=2^m (2^m-1-j)$ for some $j$.
\end{proof}

%By Theorem~\ref{thm:min-wt-is-m}, if $\#S_F=2^m$, then for all $a\in\F_{2^n}\setminus S_F$, we have \[\sum_{(j_1,j_2)\in \J_{d_1,d_2}}a^{j_1+j_2}b^{-d_1j_1-d_2j_2}=\sum_{(j_1,j_2)\in J^0_{d_1,d_2}}a^{j_1+j_2}=1. \]
We define the polynomial
\begin{equation}\label{eq:h-J-0}
	h_{J^0_{d_1,d_2}}(x)=\sum_{(j_1,j_2)\in J^0_{d_1,d_2}}x^{j_1+j_2}\in\F_2[x].
\end{equation}
Then $\deg h_{J^0_{d_1,d_2}}(x)= 2^n-2^m$ from Lemma~\ref{lem:deg} and Theorem~\ref{thm:min-wt-is-m}.
If $\ell(n)>m$, then $S_F=\F_{2^m}$ and $h_{J^0_{d_1,d_2}}(a)-1=0$ for all $a\in\F_{2^n}\setminus\F_{2^m}$ by Theorem~\ref{thm:f-x2-SF}.
Thus
\[
	h_{J^0_{d_1,d_2}}(x)-1=\frac{x^{2^n}-x}{x^{2^m}-x}=x^{2^m(2^m-1)}+x^{(2^m-1)^2}+\cdots+x^{2^m-1}+1,
\]
i.e.,
\begin{equation} \label{eq:h-expansion}
	h_{J^0_{d_1,d_2}}(x)=x^{2^m(2^m-1)}+x^{(2^m-1)^2}+x^{(2^m-1)(2^m-3)}+\cdots+
	x^{2^m-1}.
\end{equation}

\begin{corollary}\label{cor:V-structure}
	Assume that $\ell(n)>m$, $\#S_F=2^m$ and $\wt_2(d_1)=\wt_2(d_2)=2$.
	Then
	\begin{align*}
		V&:=\{j_1+j_2 :(j_1,j_2)\in J^0_{d_1,d_2}\}\\
		&\supseteq\{2^m-1,2(2^m-1),3(2^m-1),\ldots,2^m(2^m-1)\}.
	\end{align*}
	%\[j_1+j_2\in V:=\{2^m-1,2(2^m-1),3(2^m-1),\ldots,2^m(2^m-1)\}.\]
\end{corollary}
%By the definition of $V_{d_1,d_2}(j_1,j_2)$, we know that there exists a positive integer $l$ such that $l\mid n$ and
	%\begin{equation}\label{st}V_{d_1,d_2}(j_1,j_2)=V_{d_1,d_2}(2j_1,2j_2)=\cdots
	%=V_{d_1,d_2}(2^{l-1}j_1,2^{l-1}j_2).\end{equation}
% \begin{comment}
% \subsection{Properties on $d_1,d_2$}

% \rev{We next study
% 	the following modular equation
% 	\begin{equation}\label{eq:mod}
% 		d_1j_1+d_2j_2\equiv0\mod{N},
% 	\end{equation}
% 	which can be rewritten as
% 	$j_1(d_2-d_1)\equiv d_2(j_1+j_2)\mod{N}$ if $d_2>d_1$, and $j_2(d_1-d_2)\equiv d_1(j_1+j_2)\mod{N}$ if $d_2<d_1$. Thus, we next assume $d_2>d_1>0$ for simplicity, i.e.,
% 	\begin{equation} \label{eq:a}
% 		(d_2-d_1)j_1\equiv d_2(j_1+j_2)\mod{N}.
% 	\end{equation}
% }{}
% \end{comment}

% We next study the following modular equation
% \[
% 	d_1j_1+d_2j_2\equiv0\mod{N}.
% \]
% Without loss of generality, we assume that $d_2>d_1>0$.
% Then
% \begin{equation} \label{eq:a}
% 	(d_2-d_1)j_1\equiv d_2(j_1+j_2)\mod{N}.
% \end{equation}
% We let $(j_1,j_2)=\bigl(2^mj,2^m(2^m-1-j)\bigr)\in\J_{d_1,d_2}$ such that
% \begin{equation} \label{eq:a1}
% 	(d_2-d_1)j\equiv d_2(2^m-1)\mod{N}
% \end{equation}
% as given in Theorem~\ref{thm:min-wt-is-m}. Then
% \[
% 	j(d_2-d_1) \equiv0\mod{(2^m-1)}.
% \]
By Theorem~\ref{thm:min-wt-is-m}, there exists $0\le j\le 2^m-1$ such that $(j_1,j_2)=(j, 2^m-1-j)\in\J_{d_1,d_2}$ and
\begin{equation} \label{eq:a1}
	(d_2-d_1)j\equiv d_2(2^m-1)\mod{N}.
\end{equation}
Set $t:=\gcd(j,2^m-1)$.
Then we can write
\begin{equation}\label{eq:m}
	d_2-d_1=\frac{2^m-1}{t}k
\end{equation}
for some integer $k$ and Eq.~\eqref{eq:a1} becomes
\begin{equation}\label{eq:b}
	\frac{k j}{t} \equiv d_2\mod{(2^m+1)}.
\end{equation}
Set $r:=\gcd(k,2^m+1)$.
By Eqs.~\eqref{eq:m} and \eqref{eq:b}, $r$ is a common factor of $d_1-d_2$, $d_2$ and $N$.
Thus $r\mid s:=\gcd(d_1,d_2,N)$.
\begin{itemize}
	\item If $s\mid (2^m-1)$, then $r\mid \gcd(s, 2^m+1)=1$ and hence $r=1$.
	\item If $s\mid (2^m+1)$, then $s\mid k$ by Eq.~\eqref{eq:m}.
	Hence $s\mid r$ and $r=s$.
\end{itemize}
Let $u:=\gcd(t,k)$.
Then $\gcd(u,r)=1$ due to $t\mid (2^m-1)$ and $r\mid (2^m+1)$.
Hence
%Notice that $t$ is coprime to $2^m+1$, then
\begin{equation}\label{eq:gcd-d-N}
	\gcd(d_2-d_1, N)=\frac{2^m-1}{t}\gcd\bigl(k,t(2^m+1)\bigr)
	=\frac{2^m-1}{t}ur.
\end{equation}
% \[\gcd(d_1, \frac{2^m-1}{t}ur)=\gcd(d_1, d_1-d_2,N)=s. \]
% If $W:=\{a\in\F_{2^n}:a^{d_2-d_1}=1\}$, then
% \begin{equation}\label{eq:W-generator}
% 	W=\Bigl\langle\alpha^{\frac{t(2^m+1)}{ru}}\Bigr\rangle
% 	%\del{\ \text{and}\ \#W=\frac{2^m-1}{t}ur.}
% \end{equation}
% with cardinality $(2^m-1)ur/t$, where $\langle\alpha\rangle=\F_{2^n}^*$.

\begin{figure}[H]
	\includegraphics[width=.5\textwidth]{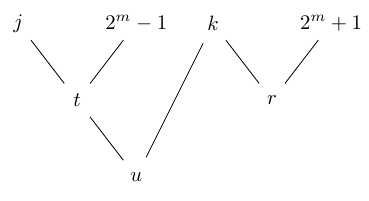}
		\caption{Division Relations of $t,k,u,r$}
\end{figure}

% \begin{comment}
% In this subsection, we will show $d_2-d_1=0\in\Z_{2^m-1}$ under the condition of Theorem~\ref{thm:f-x2-SF}.
% \rev{
% 	\begin{theorem}\label{thm:d2-d1-s-g-1}
% 		Suppose $\ell(a)>m$ for all generators $a$ of $\F_{2^n}$ and $\#S_F=2^m$. Set $s=\gcd(d_1,d_2, N)$ and $\max\{\wt_2(d_1),\wt_2(d_2)\}=2$. If $s>1$, then
% 		\[d_2-d_1\equiv0\mod{(2^m-1)}.\]
% \end{theorem}}
% {}
% \end{comment}

\begin{theorem}\label{thm:d2-d1-s-g-1}
	Assume that $\ell(n)>m$, $\#S_F=2^m$, $\wt_2(d_1)=\wt_2(d_2)=2$, and $\gcd(d_1,d_2, N)>1$.
	Then
	\[
		d_2-d_1\equiv0\mod{(2^m-1)}.
	\]
\end{theorem}

\begin{proof}
	% By Theorem~\ref{thm:f-x2-SF}, $S_F=\F_{2^m}$.
	% If $s>1$, then $s\mid (2^m+1)$ or $s\mid (2^m-1)$ by Lemma~\ref{lem:rz}.
	% Assume that $s\mid(2^m+1)$.
	% Then $W_{F_a}(0)=-2^m$ for any $a\in\F_{2^n}\setminus \F_{2^m}$  by Lemma~\ref{lem:rz}.
	% Note that
	% \begin{equation}\label{eq:sum-n-WFa0-fake1}
	% 	\begin{aligned}
	% 		\sum_{a\in\F_{2^n}}W_{F_a}(0)
	% 		&=\sum_{a\in\F_{2^n}\setminus\F_{2^m}}W_{F_a}(0)+\sum_{a\in\F_{2^m}}W_{F_a}(0)\\
	% 		&=-2^m(2^n-2^m)+\sum_{a\in\F_{2^m}}W_{F_a}(0).
	% 	\end{aligned}
	% \end{equation}
	% By Eqs.~\eqref{eq:m}, we also get
	% \begin{align*}
	% 	\sum_{a\in\F_{2^n}}W_{F_a}(0)
	% 	&=\sum_{a\in\F_{2^n}}\sum_{x\in\F_{2^n}}(-1)^{\Tr_{2^n/2}(ax^{d_1}+ax^{d_2})}\\
	% 	&=2^n+2^n\#\{x\in\F_{2^n}^*:x^{d_2-d_1}=1\}
	% \end{align*}
	% Since $x^{d_2-d_1}=1\iff x^{(2^m-1)ur/t}=1$ by Eq.~\eqref{eq:gcd-d-N}, we have
	% \begin{equation}\label{eq:sum-n-WFa0-fake2}
	% 	\sum_{a\in\F_{2^n}}W_{F_a}(0)
	% 	=2^n+2^n\cdot \frac{(2^m-1)ur}t.
	% \end{equation}
	% By Eqs.~\eqref{eq:sum-n-WFa0-fake1} and \eqref{eq:sum-n-WFa0-fake2}, we have
	% \begin{align*}
	% 	\sum_{a\in\F_{2^m}}W_{F_a}(0)
	% 	&=2^n+2^n\cdot \frac{(2^m-1)ur}t+2^m(2^n-2^m)\\
	% 	&=2^n\Bigl(\frac{2^m-1}{t}\Bigr)u\Bigl(\frac{t}{u}+r\Bigr)+2^n
	% 	\ge 2^{n+m}+2^n,
	% \end{align*}
	% which is impossible since $\sum_{a\in\F_{2^m}}W_{F_a}(0)\leq 2^{n+m}$.
	% Hence $s\nmid (2^m+1)$, this means that
	% \[
	% 	s\mid (2^m-1),\ r=1,\ \text{and}\ W_{F_a}(0)=2^m\ \text{for any}\ a\in\F_{2^n}\setminus S_F.
	% \]

	By Theorem~\ref{thm:f-x2-SF}, $S_F=\F_{2^m}$.
	If $s:=\gcd(d_1,d_2,N)>1$, then $s\mid (2^m+1)$ or $s\mid (2^m-1)$ by Lemma~\ref{lem:rz}.
	Assume that $s\mid(2^m+1)$.
	Then $W_{F_a}(0)=-2^m$ for any $a\in\F_{2^n}\setminus \F_{2^m}$  by Lemma~\ref{lem:rz}.
	Note that
	\begin{equation}\label{eq:sum-n-WFa0-fake1}
		\begin{aligned}
			\sum_{a\in\F_{2^n}}W_{F_a}(0)
			&=\sum_{a\in\F_{2^n}\setminus\F_{2^m}}W_{F_a}(0)+\sum_{a\in\F_{2^m}}W_{F_a}(0)\\
			&=-2^m(2^n-2^m)+\sum_{a\in\F_{2^m}}W_{F_a}(0).
		\end{aligned}
	\end{equation}
	By Eqs.~\eqref{eq:m}, we also get
	\begin{align*}
		\sum_{a\in\F_{2^n}}W_{F_a}(0)
		&=\sum_{a\in\F_{2^n}}\sum_{x\in\F_{2^n}}(-1)^{\Tr_{2^n/2}(ax^{d_1}+ax^{d_2})}\\
		&=2^n+2^n\#\{x\in\F_{2^n}^*:x^{d_2-d_1}=1\}
	\end{align*}
	Since $x^{d_2-d_1}=1\iff x^{(2^m-1)ur/t}=1$ by Eq.~\eqref{eq:gcd-d-N}, we have
	\begin{equation}\label{eq:sum-n-WFa0-fake2}
		\sum_{a\in\F_{2^n}}W_{F_a}(0)
		=2^n+2^n\cdot \frac{(2^m-1)ur}t.
	\end{equation}
	By Eqs.~\eqref{eq:sum-n-WFa0-fake1} and \eqref{eq:sum-n-WFa0-fake2}, we have
	\begin{equation}\label{eq:sum-m-fake}
		\sum_{a\in\F_{2^m}}W_{F_a}(0)
		=2^n+2^n\cdot \frac{(2^m-1)ur}t+2^m(2^n-2^m)
		>2^{n+m},
	\end{equation}
	which is impossible since $\sum_{a\in\F_{2^m}}W_{F_a}(0)\leq 2^{n+m}$.
	Hence $s\nmid (2^m+1)$, this means that
	\[
		s\mid (2^m-1),\ r=1,\ \text{and}\ W_{F_a}(0)=2^m\ \text{for any}\ a\in\F_{2^n}\setminus S_F.
	\]
	
	% Now Eq.~\eqref{eq:sum-n-WFa0-fake1} becomes
	% \[
	% 	\sum_{a\in\F_{2^n}}W_{F_a}(0)
	% 	=2^m(2^n-2^m)+\sum_{a\in\F_{2^m}}W_{F_a}(0).
	% \]
	% By Eq.~\eqref{eq:sum-n-WFa0-fake2}, we have
	Now Eq.~\eqref{eq:sum-m-fake} becomes
	\begin{align*}
		\sum_{a\in\F_{2^m}}W_{F_a}(0)&
		=2^n+2^n\cdot\frac{(2^m-1)t}u-2^m(2^n-2^m)\\&
		=2^n+2^n(2^m-1)\Bigl(\frac{u}{t}-1\Bigr)\le2^n.
	\end{align*}
	On the other hand,
	\begin{align*}
		\sum_{a\in\F_{2^m}}W_{F_a}(0)&=
		\sum_{x\in\F_{2^n}}\sum_{a\in\F_{2^m}}(-1)^{\Tr_{2^m/2}(a\Tr_{2^n/2^m}(x^{d_1}+x^{d_2}))}\\
		&=2^m\#\{x\in\F_{2^n} :\Tr_{2^n/2^m}(x^{d_1}+x^{d_2})=0\}\geq 2^n.
	\end{align*}
	This forces $t=u$.
	Hence $d_2-d_1=\frac ku(2^m-1)$ is divisible by $2^m-1$.
\end{proof}

%\rev{By Theorem~\ref{thm:d2-d1-s-g-1} and \ref{thm:d2-d1-s-1}, we can assume}{}
Under the assumptions in Theorem~\ref{thm:d2-d1-s-g-1}, we may assume that
\[
	d_2-d_1=(2^m-1)k,\qquad \gcd(k,2^m+1)=1.
\]
Then we have
\[
	s=\gcd(d_1,d_2,N)=\gcd(d_1,d_2-d_1,N)=\gcd(d_1,(2^m-1)k,N)=\gcd(d_1,2^m-1).
\]

\begin{remark}
  In particular, if $d_1=2^i+1$ and $d_2=2^m+1$, $0\leq i\leq m-1$, Xie et al.\cite{xie} showed that $F$ has maximal number of bent components if and only if $i=0$. However, by Theorem~\ref{thm:d2-d1-s-g-1}, we can directly obtain that $i=0$ if $F$ has maximal number of bent components and $s=\gcd(2^i+1,2^m-1)>1$.
\end{remark}

\begin{theorem}\label{2}
	Assume that $\ell(n)>m$, $\#S_F=2^m$, $\wt_2(d_1)=\wt_2(d_2)=2$, and $\gcd(d_1,d_2, N)>1$.
	Then
	% Let $F(x)=x^{d_1}+x^{d_2}:\F_{2^n}\rightarrow\F_{2^n}$ and $\#S_F=2^m$.
	% Assume that $\ell(n)>m$ and
	% If $\wt_2(d_1)=\wt_2(d_2)=2$, then
	$F(x)$ is affine equivalent to $x^{1+2^l}+x^{2^l+2^m}$,
	% $d_1=1+2^l$ and $d_2=2^l+2^m$,
	where $0<l<m$.% and $2\mid \fracm{\gcd(m,l)}$.
\end{theorem}

\begin{proof}
	We can assume $d_1=1+2^{l}$ and $d_2=2^{k_1}+2^{k_2}$ under the affine equivalent, where $0<l<m,0\leq k_1<k_2<n$ and $d_1<d_2$. By Lemma~\ref{lem:LL}, we get $\wt_2(d_2-d_1)=\wt_2((2^m-1)k)=m$.
	
	Assume that $l<k_1$.
	Then
	\[
		d_2-d_1=2^{k_1}-1+2^{k_2}-2^l
		=1+2+\cdots+2^{l-1}+2^{l+1}+\cdots+2^{k_1-1}+2^{k_2}
	\]
	and $m=\wt_2(d_2-d_1)=k_1$.
	Since $d_2-d_1=(2^m-1)+(2^{k_2}-2^l)$ is divisible by $2^m-1$, this forces $k_2=m+l$, i.e. $d_1=1+2^l,d_2=2^m+2^{m+l}$.
	This contradicts to $\#S_F=2^m$ by Lemma~\ref{pro:qua}.
	Hence $k_1\le l\le k_2$.

	Now
	\[
		d_2-d_1=1+2+\cdots+2^{k_1-1}+2^l+2^{l+1}+\cdots+2^{k_2-1}
	\]
	and $\wt_2(d_2-d_1)=k_1+k_2-l=m$, i.e., $k_2=m+l-k_1$.
	Since
	\[
		d_2-d_1=2^{m+l-k_1}+2^{k_1}-2^l-1
		=2^{l-k_1}(2^m-1)+(2^{k_1}-1)(1-2^{l-k_1}),
	\]
	we obtain that $(2^{k_1}-1)(2^{l-k_1}-1)$ is divisible by $2^m-1$.
	Since
	\[
		(2^{k_1}-1)(2^{l-k_1}-1)
		\le (2^{k_1}-1)(2^{m-1-k_1}-1)
		=2^{m-1}+1-2^{k-1}-2^{m-1-k_1}<2^m-1,
	\]
	we have $k_1=0$ or $k_1=l$.
	If $k_1=0$, then $F(x)=x^{1+2^l}+x^{1+2^{l+m}}$ is affine equivalent to $F(x^{2^{m-l}})=x^{2^{m-l}+2^m}+x^{1+2^{m-l}}$.
	If $k_1=l$, then $F(x)=x^{1+2^l}+x^{2^l+2^m}$.
\end{proof}

\section{Bounds on the nonlinearity and differential uniformity}

For a vectorial function $F:\F_{2^n}\rightarrow\F_{2^n}$, the \emph{sum-of-square-indicator} of the component function $F_u$, $u\in\F_{2^n}^*$, is defined as
\[\nu(F_u)=2^{-n}\sum_{v\in\F_{2^n}}W_{F_u}^4(v).\]
The lower bound on $\nu(F_u)$ is given by the following lemma.
\begin{lemma}[\cite{anb}] \label{lemma:anb1}
	\[
		\sum_{u\in\F_{2^n}^*}\nu(F_u)\geq (2^n-1)2^{2n+1},
	\]
and the	equality holds if and only if $F$ is APN.
\end{lemma}

We again  assume $\ell(n)>m$, $\sigma\circ F=F\circ \sigma$ and $\#~S_F=2^m$. 
By Theorem~\ref{thm:f-x2-SF} and Lemma~\ref{lemma:anb1}, we have
 \[ \begin{split}	\sum_{u\in\F_{2^m}^*}\nu(F_u)=& \sum_{u\in\F_{2^n}^*}\nu(F_u)-\sum_{u\in\F_{2^n}\setminus\F_{2^m}}\nu(F_u)\\	
 	\geq & (2^n-1)2^{2n+1}-2^{2n}(2^n-2^m)=2^{2n}(2^n+2^m-2). 	
 \end{split} \]
By the definition of $\nu(F_u)$, we know
\begin{align*}
	\sum_{u\in\F_{2^m}^*}\nu(F_u)
	&\le 2^{-n}\max_{u\in\F_{2^m}^*,v\in\F_{2^n}}
	W^2_{F_u}(v)\sum_{u\in\F_{2^m}^*}\sum_{v\in\F_{2^n}}W_{F_u}^2(v)\\
	&=2^{n}(2^m-1)\max_{u\in\F_{2^m}^*,v\in\F_{2^n}}
	W^2_{F_u}(v).
\end{align*}
Therefore, %by Eq.~\eqref{eq:nu},
we have
\begin{equation}\label{eq:lower-bound-max}
	\max_{u\in\F_{2^m}^*,v\in\F_{2^n}}
	W^2_{F_u}(v)\geq\frac{2^{2n}(2^n+2^m-2)}{2^{n}(2^m-1)}= 2^n(2^m+2).
\end{equation}

\begin{theorem}\label{thm:bd1}
	Assume that $\ell(n)>m$,  $\sigma\circ F=F\circ \sigma$ and $\#S_F=2^m$. Then the nonlinearity of $F(x)$ is
	\[
		\N_F\leq 2^{n-1}-\frac{1}{2}\sqrt{2^n(2^m+2)}.
	\]
	Furthermore, if $F$ is plateaued, then $\N_F\leq 2^{n-1}-2^{\lfloor\frac{3n}{4}\rfloor}$.
\end{theorem}

\begin{proof}
	The bound of $\N_F$ is directly obtained by Eq.~\eqref{eq:lower-bound-max}.
	If $F$ is plateaued, we can take
	\[
		\max_{u\in\F_{2^m}^*,v\in\F_{2^n}}
		W^2_{F_u}(v)=2^{n+k}
	\]
	for some even $k>m$ by Definition~\ref{def:vectorial-function}(iv).
	By Eq.~\eqref{eq:lower-bound-max}, we have $2^k\geq 2^m+2$, which implies that $k\geq 2\lfloor\frac m2\rfloor+2=2\lfloor\frac n4\rfloor+2$.
	% if $2\nmid m$; $k\geq m+2$ if $2\mid m$.
	Hence
	\[
		\N_F
		\le 2^{n-1}-\frac12\cdot 2^{m+\lfloor\frac n4\rfloor+1}
		\leq 2^{n-1}-2^{\lfloor\frac{3n}{4}\rfloor}.\qedhere
	\]
\end{proof}

\begin{remark}
	The bound in Theorem~\ref{thm:bd1} is still true
	 for non-plateaued functions, but it may be weaker than the bound given in \cite[Conjecture~1]{xie}.
\end{remark}

\begin{lemma}\label{lem:bound-same-image}
Let $F:\F_{2^n}\rightarrow \F_{2^n}$. Then 
	\[
		\#\{(x,y)\in\F_{2^n}\times\F_{2^n}:F(x)=F(y)\}
		\ge\frac{2^{2n}}{\#\Im(F)}.
	\]
\end{lemma}
\begin{proof}
For any $b\in\F_{2^n}$, denote by $\#F^{-1}(b)$ the size of the pre-image of $b$ by $F$.
By the Cauchy-Schwarz inequality,  we have
	\[
		\text{LHS}
		=\sum_{b\in\Im(F)}\bigl(\#F^{-1}(b)\bigr)^2
		\ge \frac{\biggl(\sum\limits_{b\in\Im(F)}\#F^{-1}(b)\biggr)^2}{\#\Im(F)}
		=\frac{2^{n}}{\#\Im(F)}.
		\qedhere
	\]
\end{proof}

We now give a bound which depends on the cardinality of the image set of $F$.

\begin{theorem}\label{bd2}
	Assume that $\ell(n)>m$,  $\sigma\circ F=F\circ \sigma$ and $\#S_F=2^m$.
	Set
	\[
		T:=\#\{u\in\F_{2^m}^*:W_{F_u}(0)\neq0\}.
	\]
	If $T\neq0$, then
	
	\[
		\N_F
		\le 2^{n-1}-\frac{1}{2}\sqrt{\frac{2^{3m}}T\biggl(\frac{2^{3m}}{\#\Im(F)}-2^{m+1}+1\biggr)}.
	\]
\end{theorem}
\begin{proof}
	By Lemma~\ref{lem:bound-same-image}, we have
	\[
		\sum_{u\in\F_{2^n}}W_{F_u}^2(0)
		=2^n\#\{(x,y)\in\F_{2^n}\times\F_{2^n}:F(x)=F(y)\}
		\ge\frac{2^{3n}}{\#\Im(F)}.
	\]
	Meanwhile,

	\[
		\sum_{u\in\F_{2^n}}	W_{F_u}^2(0)
		=\sum_{u\in\F_{2^n}\setminus\F_{2^m}}W_{F_u}^2(0)+\sum_{u\in\F_{2^m}}W_{F_u}^2(0)
		=2^n(2^n-2^m)+\sum_{u\in\F_{2^m}}W_{F_u}^2(0).
	\]
	Thus
	\[
		\sum_{u\in\F_{2^m}^*}W_{F_u}^2(0)\geq \frac{2^{3n}}{\#\Im(F)}-2^{2n+1}+2^{n+m}
	\]
	and
	\begin{align*}
		\max_{u\in\F_{2^m}^*,v\in\F_{2^n}}W^2_{F_u}(v)&
		\ge\max_{u\in\F_{2^m}^*}W^2_{F_u}(0)
		\ge\frac1T\sum_{u\in\F_{2^m}^*}W_{F_u}^2(0)\\&
		\ge\frac1T\biggl(2^n(2^n-2^m)+\sum_{u\in\F_{2^m}}W_{F_u}^2(0)\biggr).
	\end{align*}
We thus get the desired bound for $\N_F$.
\end{proof}

\begin{remark}
Clearly $T\leq 2^m-1$, hence Theorem~\ref{bd2} implies that
\[
		\N_F\leq 2^{n-1}-\frac{1}{2}\sqrt{\frac{2^{3m}}{2^m-1}\biggl(\frac{2^{3m}}{\#\Im(F)}-2^{m+1}+1\biggr)},
	\]
	which is smaller than Carlet's bound $\displaystyle 2^{n-1}-\sqrt{\frac1{2^n-1}\biggl(\frac{2^{3n-2}}{\#\Im(F)}-2^{2n-2}\biggr)}$ in \cite[Proposition~2]{car}.
	Moreover, this bound is also smaller than $2^{n-1}-\frac{1}{2}\sqrt{2^n(2^m+2)}$ given by Theorem~\ref{thm:bd1} when $\#\Im(F)\leq \frac{2^{3n}}{3\cdot2^{2n}-2^{n+1}}$.
\end{remark}

The following bound can be found in \cite{car}, we will provide a whole proof for completeness.

\begin{theorem}\label{thm:bound-delta}
	For any non-injective vectorial function $F:\F_{2^n}\rightarrow\F_{2^n}$, we have
	\[
		\delta_F\geq \left\lceil\frac{2^n}{\#\Delta}\biggl(\frac{2^n}{\#\Im(F)}-1\biggr)\right\rceil,
	\]
	where
	\[
		\Delta=\{x+y:(x,y)\in\F_{2^n}\times\F_{2^n},x\neq y,F(x)=F(y)\}.
	\]
\end{theorem}

\begin{proof}
	Denote by
	\[
		\delta_{a,b}:=\#\{x\in\F_{2^n} : F(x+a)+F(x)=b\}.
	\]
	Then for any nonzero $a\notin\Delta$, $\delta_{a,0}=0$.
	By Definition~\ref{def:vectorial-function} (vi), we have
	\[
		\delta_F
		\geq\max_{a\in\F_{2^n}^*}\delta_{a,0}
		\geq\frac{\sum_{a\in\F_{2^n}^*}\delta_{a,0}}{\#\Delta}=\frac{\{(x,y)\in\F_{2^n}\times\F_{2^n}:F(x)=F(y)\}-2^n}{\#\Delta}.
	\]
	The bound then can be obtained by Lemma~\ref{lem:bound-same-image}.
\end{proof}
\begin{remark}
	From the definition of $\Delta$, if $\ell(n)>m$, $\sigma\circ F=F\circ \sigma$ and $\#S_F=2^m$, then
	\begin{align*}
		\#\Delta & \leq \frac{1}{2}\#\{(x,y)\in\F_{2^n}\times\F_{2^n} : F(x)=F(y)\}-2^{n-1}\\
		&=\frac{1}{2^{n+1}}\sum_{x,y,a\in\F_{2^n}}(-1)^{\Tr_{2^n/2}(a(F(x)+F(y)))}
		-2^{n-1}\\
		&=\frac{1}{2^{n+1}}\sum_{a\in\F^*_{2^n}}W_{F_a}^2(0)
		=\frac{1}{2^{n+1}}\sum_{a\in\F^*_{2^m}}W_{F_a}^2(0)+2^{n-1}-2^{m-1}.
	\end{align*}
	In \cite[Theorem~7]{PPMB2018}, we know that $W_{F_a}(0)=0$ for all $a\in\F_{2^m}^*$ if $s=\gcd(d_1,2^m-1)=1$.
	Then $\#\Delta\leq 2^{n-1}-2^{m-1}$ and
	% \begin{equation}\label{bde}
	\[
		\delta_F\geq\left\lceil\frac{2^n(2^n-\#\Im(F))}{(2^{n-1}-2^{m-1})\#\Im(F)}\right\rceil\geq \left\lceil\frac{2^{n+1}}{\#\Im(F)}\right\rceil-2.
	\]
% \end{equation}
\end{remark}

We next give the cardinality of the image set of binomial vectorial function $F(x)=x^{d_1}+x^{d_2}\in\F_{2^n}[x]$.

\begin{theorem}\label{thm:size-image-F}
Suppose $\ell(n)>m$, $\sigma\circ F=F\circ \sigma$, $\#S_F=2^m$ and $\wt_2(d_1)=\wt_2(d_2)=2$. Set $s=\gcd(d_1,d_2, N)>1$ and $\F_{2^n}^*=\langle\alpha\rangle$.
	Then
	\[
		\#\Im(F)=\frac{(2^m-1)c}{s}+1,
	\]
	where $c=\#\{F(\alpha^i)^{(2^m-1)/s}:\ i=1,2,\ldots,2^m\}$.
\end{theorem}

\begin{proof}
	By Theorems~\ref{thm:f-x2-SF},\ref{thm:min-wt-is-m} and \ref{thm:d2-d1-s-g-1}, we have $d_2-d_1=(2^m-1)k$, $\gcd(k,2^m+1)=1$, and then $\Ker(F)=\F_{2^m}$.
	If $x\in\alpha^i \F_{2^m}^*$, then
	\[
		F(x)=x^{d_1}(1+x^{d_2-d_1})
		=x^{d_1}(1+x^{(2^m-1)k})
		=x^{d_1}f(\alpha^{(2^m-1)i}),
	\]
	where $f(x)=1+x^k$.
	Let
	\[
		D=\langle\alpha^{(2^m+1)s}\rangle\subseteq\F_{2^m}^*
	\]
	be the image of $x^{d_1}:\F_{2^m}^*\mapsto D$.
	Then $x^{d_1}$ maps each coset $\alpha^i\F_{2^m}^*$ onto $\alpha^{d_1i}D$, and the image of the map
	\[
		F:\alpha^i\F_{2^m}^*\mapsto \alpha^{d_1i}f(\alpha^{(2^m-1)i})D=F(\alpha^i)D
	\]
	is the set $\{0\}$ or a coset of $D$.
	Note that $F(\alpha^i)D=F(\alpha^j)D$ if and only if $F(\alpha^i)^{(2^m-1)/s}=F(\alpha^j)^{(2^m-1)/s}$. Hence,
	\[
		\#\Im(F)=1+c\#D=1+\frac{(2^m-1)c}{s}.
	\]
	Therefore, the proof is done.
\end{proof}

In particular, for $d_1=1,d_2=2^{m-1}(2^m+1)$ and $d_1=2^l+1,d_2=2^m+2^l$, we can determine $\#\Im(F)$ explicitly.

\begin{theorem}
	The image of $F(x)=x^{2^l+1}+x^{2^l+2^m}$ has size
	\[
		\#\Im(F)=1+\frac{2^n-2^m}{\gcd(2^l+1,2^m-1)}
		=\begin{cases}
			2^n-2^m+1,&\text{if}\ v_2(m)\le v_2(l);\\
			(2^n-2^m)/3+1,&\text{if}\ v_2(m)>v_2(l).
		\end{cases}
	\]
	In particular, the image of $F(x)=x+x^{2^{m-1}(2^m+1)}$ has size $\#\Im(F)=2^n-2^m+1$.
\end{theorem}

\begin{proof}
	Denote by $\Tr=\Tr_{2^n/2^m}$.
	Then $F(x)=x^{2^l}\Tr(x)$.
	If $x\in\F_{2^m}$, then $\Tr(x)=0$ and $F(x)=0$.

	If $F(x)=y\notin\F_{2^m}$, then $x\ne0,\Tr(x)=x^{-2^l}y$ and
	\[
		\Tr(y^{2^{n-l}})
		=\Tr\bigl(x\Tr(x)^{2^{n-l}}\bigr)
		=\Tr(x)^{1+2^{n-l}}
	\]
	belongs
	\[
		D:=\{u^{1+2^{n-l}}:u\in\F_{2^m}^*\}\subseteq \F_{2^m}^*.
	\]
	If $\Tr(y^{2^{n-l}})=u^{1+2^{n-l}}$, then $\Tr(x)=u\zeta$ for some $\zeta$ in
	\[
		H:=\{\zeta\in\F_{2^m}^*:\zeta^{1+2^{n-l}}=1\}\subseteq \F_{2^m}^*,
	\]
	which has size
	\[
		\#H=\gcd(2^{n-l}+1,2^m-1)
		=\gcd(2^l+1,2^m-1)
		=\begin{cases}
			1,&\text{if}\ v_2(m)\le v_2(l);\\
			3,&\text{if}\ v_2(m)>v_2(l).
		\end{cases}
	\]
	Thus
	\[
		x=\biggl(\frac{y}{\Tr(x)}\biggr)^{2^{n-l}}
		=\biggl(\frac{y}{u\zeta}\biggr)^{2^{n-l}}
		=\frac{y^{2^{n-l}}}{u^{2^{n-l}}\zeta}
	\]
	lies in the pre-image of $y$, which has $\#H$ possible values.
	Hence
	\[
		\#\Im(F)=1+\frac{2^n-2^m}{\#H}.
	\]
	If $l=0$, then $F$ is affine equivalent to $x+x^{2^{m-1}(2^m+1)}$, whose image has size $2^n-2^m+1$.
\end{proof}

By Theorems~\ref{bd2}-\ref{thm:size-image-F}, we have

\begin{theorem}\label{thm:imgbd}
Adopt  the same notations as Theorem~\ref{thm:size-image-F}.
	Then the nonlinearity and differential uniformity of $F$ satisfy
	\begin{align*}
		\N_F&\leq 2^{n-1}-\frac{1}{2}\sqrt{\frac{2^{3m}}{T}\cdot\biggl(\frac{2^{3m}s}{s+(2^m-1)c}-2^{m+1}+1\biggr)},\\
		\delta_F&\geq\ceil{\frac{2^n}{\#\Delta}\biggl(\frac{2^ns}{s+(2^m-1)c}-1\biggr)}.
	\end{align*}
\end{theorem}

\begin{remark}
	% From Theorem~\ref{thm:imgbd}, we get
	Since $s=\gcd(d_1,2^m-1)\geq3$ and
	\[
		s+(2^m-1)c\leq s+2^m(2^m-1)<(2^m+1)(2^m-1)=2^n-1,
	\]
	we have
	\begin{align*}
		\frac{2^{3m}}{T}&\cdot\biggl(\frac{2^{3m}s}{s+(2^m-1)c}-2^{m+1}+1\biggr)\\
		>{}&\frac{2^{3m}}{2^m-1}\biggl(\frac{2^{3m}s}{2^n-1}-2^{m+1}+1\biggr)
		>\frac{2^{2n}(s-2)+2^{n+m}}{2^m-1}\\
		\ge{}&\frac{2^{2n}+2^{n+m}}{2^m-1}>2^n(2^m+2).
	\end{align*}
	% \[
	% 	\frac{\frac{2^{3n}s}{(2^m-1)c+s}-2^{2n+1}+2^{n+m}}{T}>
	% 	\frac{\frac{2^{3n}s}{(2^m+1)(2^m-1)}-2^{2n+1}+2^{n+m}}{2^m-1}>
	% 	\frac{2^{2n}(s-2)+2^{n+m}}{2^m-1}
	% \]
	% and $s=\gcd(d_1,2^m-1)\geq3$.
	Thus the bound of $\N_F$ in Theorem~\ref{thm:imgbd} is smaller than $2^{n-1}-\frac{1}{2}\sqrt{2^n(2^m+2)}$.
	%From Theorem~\ref{thm:pott},  $s=\gcd(2^i+1,2^m-1)=1$ if $2\nmid \fracm{\gcd(m,i)}$, or $=2^{\gcd(m,i)}+1$ if $2\mid \fracm{\gcd(m,i)}$.
	Moreover, we get
	\[
		\delta_F
		\geq\ceil{\frac{2^ns}{s+(2^m-1)c}-1}
		\geq\ceil{\frac{2^ns}{2^n-1}-1}
		\ge s-1.
	\]
	% \geq\left\lceil\frac{2^{2n}s-c2^n(2^m-1)}{c(2^m-1)(2^n-1)}\right\rceil\geq (s-1).\]
\end{remark}

\section{Conclusion}
In this paper, we use Stickelberger's Theorem to study the Walsh transform of binomial vectorial function $F(x)=x^{d_1}+x^{d_2}$, and show that  $F(x)$ is equivalent to $x^{2^m+1}$ when $\wt_2(d_1)=1$, and to $x^{2^i+1}+x^{2^m+2^i}$ when $\wt_2(d_1)=\wt_2(d_2)=2$ under a technical condition, where $0<i\leq m-1$. Moreover, we give the size of the image set of $F$, and then give the bounds on the nonlinearity and differential uniform of $F(x)$ by means of its image set size.

\end{document}